\documentclass[aps,prl,reprint,superscriptaddress,nofootinbib,nobibnotes]{revtex4-2}

\usepackage{amsmath,amssymb,amsthm}
\usepackage{graphicx}
\usepackage[colorlinks=true,allcolors=blue]{hyperref}

\newtheorem{theorem}{Theorem}

\newtheorem{lemma}{Lemma}
\newtheorem{corollary}{Corollary}
\newtheorem{definition}{Definition}

\newcommand{\Hmat}{\mathbf{H}}
\newcommand{\Smat}{\mathbf{S}}
\newcommand{\tHmat}{\widetilde{\mathbf{H}}}
\newcommand{\tSmat}{\widetilde{\mathbf{S}}}
\newcommand{\dHmat}{\Delta\mathbf{H}}
\newcommand{\dSmat}{\Delta\mathbf{S}}
\newcommand{\Nshots}{N_{\rm shots}}

\newcommand{\mri}{\mathrm{i}}
\newcommand{\vareps}{\varepsilon}

\begin{document}

\title{Improved Measurement Cost Scaling in the Nonorthogonal Quantum Eigensolver}

\author{Mingyu Kang}
\affiliation{Berkeley Center for Quantum Information and Computation, Berkeley, California 94720, USA}
\affiliation{Department of Chemistry, University of California, Berkeley, California 94720, USA}
\affiliation{Challenge Institute for Quantum Computation, Berkeley, California 94720, USA}
\author{K. Birgitta Whaley}
\affiliation{Berkeley Center for Quantum Information and Computation, Berkeley, California 94720, USA}
\affiliation{Department of Chemistry, University of California, Berkeley, California 94720, USA}
\affiliation{Challenge Institute for Quantum Computation, Berkeley, California 94720, USA}

\date{\today}

\begin{abstract}
Quantum subspace diagonalization methods are promising algorithms for quantum
chemistry on near-term quantum computers. These methods can estimate low-lying
energies of molecular systems using shallow quantum circuits, at the cost of
many circuit repetitions to estimate the projected matrix elements. Errors in
these matrix elements can be converted into much larger eigenvalue errors by an
ill-conditioned overlap matrix. We study this bottleneck for the nonorthogonal
quantum eigensolver (NOQE), which constructs a compact multireference subspace
from dressed unrestricted Hartree-Fock states.
We prove a finite-shot perturbation bound showing that, after overlap
thresholding, the eigenvalue sensitivity is controlled by the condition number
of the retained overlap matrix rather than by a worst-case dimension factor.
With a scalable thresholding scheme, the upper bound on the per-matrix-element shot count 
required to reach a target accuracy scales as $\mathcal{O}(M)$, improving on the 
previously known $\mathcal{O}(M^3)$ bound, where $M$ is the number of reference states. 
Numerical experiments on hydrogen chains and rings suggest that, in practice,
the measurement cost of structured NOQE instances can grow even more slowly
than this linear bound.
\end{abstract}

\maketitle

Estimating low-lying energies of molecular Hamiltonians is one of the most widely
anticipated applications of quantum computers~\cite{mcardle2020quantum}.
Quantum phase estimation~\cite{kitaev1995quantum,abrams1999quantum} approaches provide a direct route to this goal, but
require large fault-tolerant quantum computers~\cite{vonburg2021quantum,lee2021even,low2025fast}. Quantum
subspace diagonalization methods are a promising alternative for near-term and
early fault-tolerant devices, as they use much shallower circuits than the long
coherent circuits required to prepare an eigenstate directly.

A quantum subspace diagonalization method defines a collection of trial states,
uses shallow circuits to estimate the Hamiltonian and overlap matrix elements
between them, and solves the resulting generalized eigenvalue problem on a
classical computer. Representative ways of constructing the subspace include
quantum subspace expansion~\cite{mcclean2017hybrid,colless2018computation,patel2026quantum},
quantum Krylov methods~\cite{stair2020multireference,cortes2022fast,shen2023real},
quantum power methods~\cite{seki2021quantum}, quantum filter diagonalization~\cite{parrish2019quantum,bespalova2021hamiltonian}, and the tensor network quantum eigensolver~\cite{leimkuhler2025quantum,leimkuhler2025exponential}.
The use of shallow circuits comes at the cost of a measurement
problem: each entry of the projected Hamiltonian and overlap matrices is
estimated from repeated measurements. Small errors in individual matrix
elements can produce much larger errors in the generalized eigenvalues after
the classical diagonalization step. Existing perturbation analyses show 
that this element-to-eigenvalue error conversion can be severe
when the overlap matrix is ill-conditioned~\cite{mathias2004definite,epperly2022theory,lee2024sampling,kirby2024analysis}.
Thus, one may need very many shots to estimate the matrix elements accurately enough
that the final eigenvalue error remains below a target tolerance.

The nonorthogonal quantum eigensolver (NOQE) is a leading quantum subspace
diagonalization method for quantum chemistry~\cite{baek2023say}. NOQE is promising as its
reference states are designed to capture both static and dynamic correlations. Allowing the references to
be nonorthogonal gives the design freedom to use unrestricted Hartree-Fock
(UHF) states, which can improve the accuracy of ground-state energy estimates.
At the same time, large overlaps between reference states can make the overlap
matrix ill-conditioned and blow up the measurement cost. We show theoretical
and numerical evidence that this need not be an issue for NOQE when
thresholding is done properly.

This work develops a perturbation-theory framework for the finite-shot
stability of thresholded NOQE. We first review the NOQE algorithm and show how
finite-shot estimates of the Hamiltonian and overlap matrices define a
perturbed generalized eigenvalue problem. We then prove that the eigenvalue sensitivity is controlled
by the condition number of the thresholded overlap matrix, rather than by a
worst-case dimension factor. This reduces the sufficient per-matrix-element
shot count from the previously known $\mathcal{O}(M^3)$ scaling~\cite{lee2024sampling,ren2026error} to
$\mathcal{O}(M)$, where $M$ is the number of reference states, under a scalable
thresholding scheme. Finally, hydrogen-chain
and hydrogen-ring numerics support this reduction by showing that the
condition number of the thresholded overlap matrix remains controlled, as well as 
suggest that the measurement cost scales even more favorably in practice.

\emph{NOQE algorithm.} ---
NOQE~\cite{baek2023say} constructs a compact multireference ansatz in the subspace spanned by $M$ dressed
nonorthogonal states, hereafter referred to as \textit{reference states}. 
The starting point for the $J$th state is a UHF solution
$\lvert \Phi_J^{\rm UHF}\rangle$, obtained as a distinct stationary solution of
the self-consistent-field equations by optimizing from a particular localized 
electronic configuration as the initial guess. The NOQE reference state
$\lvert \phi_J\rangle$ is then constructed by applying a shallow correlating
operator to this UHF state,
\[
    \lvert \phi_J\rangle =
    \hat{U}_J \lvert \Phi_J^{\rm UHF}\rangle .
\]
In this construction, the UHF states capture static correlation, while the
dressing operators introduce dynamic correlation.
The dressing operator $\hat{U}_J$ is chosen from compact unitary
coupled-cluster-like~\cite{romero2018strategies,lee2019generalized} or
cluster-Jastrow~\cite{matsuzawa2020jastrow,tkachenko2025beyond} correlators, with parameters
supplied classically rather than optimized on the quantum device.
NOQE approximates a low-lying eigenstate within the span of these reference
states,
\[
    \lvert \Psi\rangle =
    \sum_{J=1}^M c_J \lvert \phi_J\rangle ,
\]
where the coefficients $c_J$ are determined after projecting the Hamiltonian
and overlap operators onto this multireference subspace and diagonalizing the
resulting generalized eigenvalue problem. 
Specifically, the projected Hamiltonian and overlap matrices are
\[
    H_{IJ}=\langle \phi_I|\hat{H}|\phi_J\rangle
    \qquad \text{and} \qquad
    S_{IJ}=\langle \phi_I|\phi_J\rangle,
\]
such that the coefficients and energies satisfy 
\begin{equation}
    \Hmat \vec{c} = E \Smat \vec{c}
    \label{eq:noqe-gevp}
\end{equation}
This generalized eigenvalue problem is solved on a classical computer, 
which is feasible as the $M$-dimensional
diagonalization is small compared with representing the Hamiltonian in the full
$2^{N_{\rm SO}}$-dimensional Fock space, where $N_{\rm SO}$ is the number of
spin orbitals. 

The numerical stability of Eq.~\eqref{eq:noqe-gevp}
depends on the conditioning of the overlap matrix. If $\Smat$ is nearly
singular, thresholding needs to be performed to remove directions with small
overlap-matrix eigenvalues before the generalized eigenvalue problem is solved.

The quantum computer is used to estimate the matrix elements, where the off-diagonal ones are
classically expensive for dressed nonorthogonal UHF states. The matrix elements are
measured using a modified-Hadamard-test protocol~\cite{huggins2020nonorthogonal}
that uses two $N_{\rm SO}$-qubit registers and one control qubit. A variant with a single 
$N_{\rm SO}$-qubit system register and one control qubit is also possible, 
at the cost of approximately doubling the circuit depth~\cite{baek2023say}.
As different UHF states are naturally expressed in different molecular-orbital
bases, these circuits include orbital-basis rotations that map the references
to a common qubit-basis representation~\cite{baek2023say}.

Finite circuit repetitions produce noisy estimates of the matrix elements,
which perturb the generalized eigenvalue problem. The central measurement-cost
question is how many shots per matrix element are required to keep the induced
eigenvalue error below a target tolerance.

\emph{Theory.} ---
Let $\Hmat$ and $\Smat$ denote the exact projected Hamiltonian and overlap
matrices. We write the errors from finite-shot measurement as $\dHmat$ and
$\dSmat$, such that the measured projected matrices are
\begin{equation}
    \tHmat=\Hmat+\dHmat,\qquad
    \tSmat=\Smat+\dSmat .
\end{equation}

We denote the size of the perturbation by
\begin{equation}
    \chi=\sqrt{\|\dHmat\|^2+\|\dSmat\|^2}.
    \label{eq:chi-def}
\end{equation}
Assuming unbiased finite-shot measurement, $\dHmat$ and $\dSmat$ are modeled as
$M\times M$ random matrices where each element is drawn from zero-mean Gaussian 
distribution with variance $\sigma^2=\mathcal{O}(1/\Nshots)$, where $\Nshots$ denotes the number of circuit repetitions, or shots, used per matrix
element. 
Then, random-matrix concentration gives, with high probability~\cite{vershynin2012nonasymptotic},
\begin{equation}
    \chi
    =
    \mathcal{O}\!\left(\sqrt{{M} / \Nshots}\right).
    \label{eq:chi-rmt-scaling}
\end{equation}
We note that there is a constant suppressed in the asymptotic scaling bound of Eq.~\eqref{eq:chi-rmt-scaling} that is related to the variance of the Hamiltonian. The precise form and value of this depend on the
Pauli decomposition of the Hamiltonian, $\hat{H}=\sum_l w_l \hat{P}_l$, and on
the measurement strategy, e.g., employing a Pauli-word
grouping~\cite{verteletskyi2020measurement,yen2023deterministic,shlosberg2023adaptive} or a basis-rotation grouping~\cite{huggins2021efficient}. This can
introduce additional system-size-dependent measurement overheads on the variance $\sigma^2$. In this work, we focus
on the scaling with subspace dimension $M$ and do not model these
Hamiltonian- and grouping-dependent prefactors.
We also note that iterative quantum amplitude estimation can achieve the
quadratically improved Heisenberg scaling
$\sigma^2=\mathcal{O}(1/\Nshots^2)$ when $\Nshots$ counts coherent oracle
queries rather than independent projective
measurements~\cite{grinko2021iterative,ren2026towards}.

Mathias and Li~\cite{mathias2004definite} turn the generalized eigenvalue problem into a
geometric problem in a two-dimensional plane~\cite{epperly2022theory,lee2024sampling}.
Let $E_j$ denote the exact projected generalized eigenvalues, ordered in
increasing value. The equation
$\Hmat \vec{c}_j=E_j\Smat\vec{c}_j$ can be written as
$\beta_j\Hmat\vec{c}_j=\alpha_j\Smat\vec{c}_j$, with
$E_j=\alpha_j/\beta_j$, such that each generalized eigenvalue is represented by a point
$(\alpha_j,\beta_j)$. The corresponding eigenangle is
\[
    \theta_j=\tan^{-1}(\beta_j/\alpha_j)=\tan^{-1}(1/E_j).
\]
The numerical range of
$\Hmat+\mri\Smat$ provides the geometric plane in which these angles are
compared. Note that the normalization of $(\alpha_j,\beta_j)$ is arbitrary; Mathias and
Li fix this freedom by using a diagonalizing matrix with unit columns, denoted
by $X$ below.

Let $\widetilde{E}_j$ and $\widetilde{\theta}_j$ denote the corresponding
eigenvalues and eigenangles of the perturbed pair $(\tHmat,\tSmat)$. The
target is to choose $\Nshots$ such that, for the low-lying eigenvalue of
interest, the eigenangle error $|\widetilde{\theta}_j-\theta_j|$ is at most a
prescribed tolerance $\mathcal{E}$ with high probability. We choose to state
the target as an eigenangle error because it is the natural error metric of
the geometric formulation, which yields the following perturbation bound derived by Mathias and Li~\cite{mathias2004definite}.
\begin{theorem}[Mathias-Li bound~\cite{mathias2004definite}]
\label{thm:mathias-li-bound}
Let $(\Hmat,\Smat)$ be a definite Hermitian pair with $\Smat>0$, and let
$(\tHmat,\tSmat)$ be its perturbed pair. Let $X$ be the unit-column
matrix such that
\[
    X^\dagger(\Hmat+\mri\Smat)X
    =
    \operatorname{diag}(\alpha_1+\mri\beta_1,\ldots,\alpha_M+\mri\beta_M).
\]
Define
\[
    d_{\min}=\min_j |\alpha_j+\mri\beta_j|.
\]
If the perturbation is sufficiently small such that $\chi\|X\|^2<\min_j \beta_j$, 
then the generalized eigenangle error obeys
\begin{equation}
    |\widetilde{\theta}_j-\theta_j|
    \leq
    \sin^{-1}\!\left(
        \frac{\chi\|X\|^2}{d_{\min}}
    \right).
    \label{eq:angle-bound}
\end{equation}
\end{theorem}

\begin{proof}
See Ref.~\cite{mathias2004definite}.
\end{proof}

We now show in Corollary 1 that combining Theorem~\ref{thm:mathias-li-bound} with Eq.~\eqref{eq:chi-rmt-scaling} 
gives the best known upper bound on
the per-element shot count, which scales as $\mathcal{O}(M^3)$~\cite{lee2024sampling,ren2026error}.

\begin{corollary}[Best known shot-cost upper bound]
\label{cor:worst-case-shot-cost}
Assuming the random-matrix scaling Eq.~\eqref{eq:chi-rmt-scaling}, the number
of shots per matrix element
sufficient to make the finite-shot eigenangle error at most
$\mathcal{E}$ is
\begin{equation}
    \Nshots =
    \mathcal{O}\!\left(\frac{M^3}{d_{\min}^2\mathcal{E}^2}\right).
    \label{eq:worst-case-shot-scaling}
\end{equation}
\end{corollary}

\begin{proof}
For a Hermitian matrix $A$, let $\lambda_{\max}(A)$ and $\lambda_{\min}(A)$
be its largest and smallest eigenvalues.
Since $X$ has unit columns,
$\|X\|^2=\lambda_{\max}(X^\dagger X)\leq\operatorname{Tr}(X^\dagger X)=M$.
Theorem~\ref{thm:mathias-li-bound} therefore gives
$|\widetilde{\theta}_j-\theta_j|
=\mathcal{O}(\chi M/d_{\min})$. Substituting Eq.~\eqref{eq:chi-rmt-scaling}
and requiring this eigenangle error to be at most
$\mathcal{E}$ gives
Eq.~\eqref{eq:worst-case-shot-scaling}.
\end{proof}

The source of the severe $M$ dependence in
Corollary~\ref{cor:worst-case-shot-cost} is the dimension bound
$\|X\|^2\leq M$. This bound ignores the overlap geometry of the nonorthogonal
reference states. The following lemma replaces the dimension factor by the
overlap condition number; after thresholding is introduced below, this is the
step that will improve the measurement-cost upper bound.

\begin{lemma}[Overlap condition-number bound]
\label{lem:x-condition-number}
In the setting of Theorem~\ref{thm:mathias-li-bound}, the unit-column matrix
$X$ satisfies
\begin{equation}
    \|X\|^2 \leq \kappa(\Smat).
    \label{eq:x-condition-number}
\end{equation}
Here $\kappa(\Smat)=\lambda_{\max}(\Smat)/\lambda_{\min}(\Smat)$ is the
condition number of the positive definite overlap matrix.
\end{lemma}
\begin{proof}
See End Matter.
\end{proof}

Replacing the worst-case factor $M$ by $\kappa(\Smat)$ in the upper bound is
useful only if $\kappa(\Smat)$ grows sublinearly with
$M$. 
Such conditioning is not guaranteed for the raw overlap matrix, which can
become nearly singular as more nonorthogonal reference states are added.
Practical NOQE calculations address this by thresholding the overlap matrix:
directions with small overlap eigenvalues are discarded before the generalized
eigenvalue problem is solved~\cite{baek2023say}. We formalize this thresholded setting next.

\begin{figure*}[t]
\centering
\includegraphics[width=\textwidth]{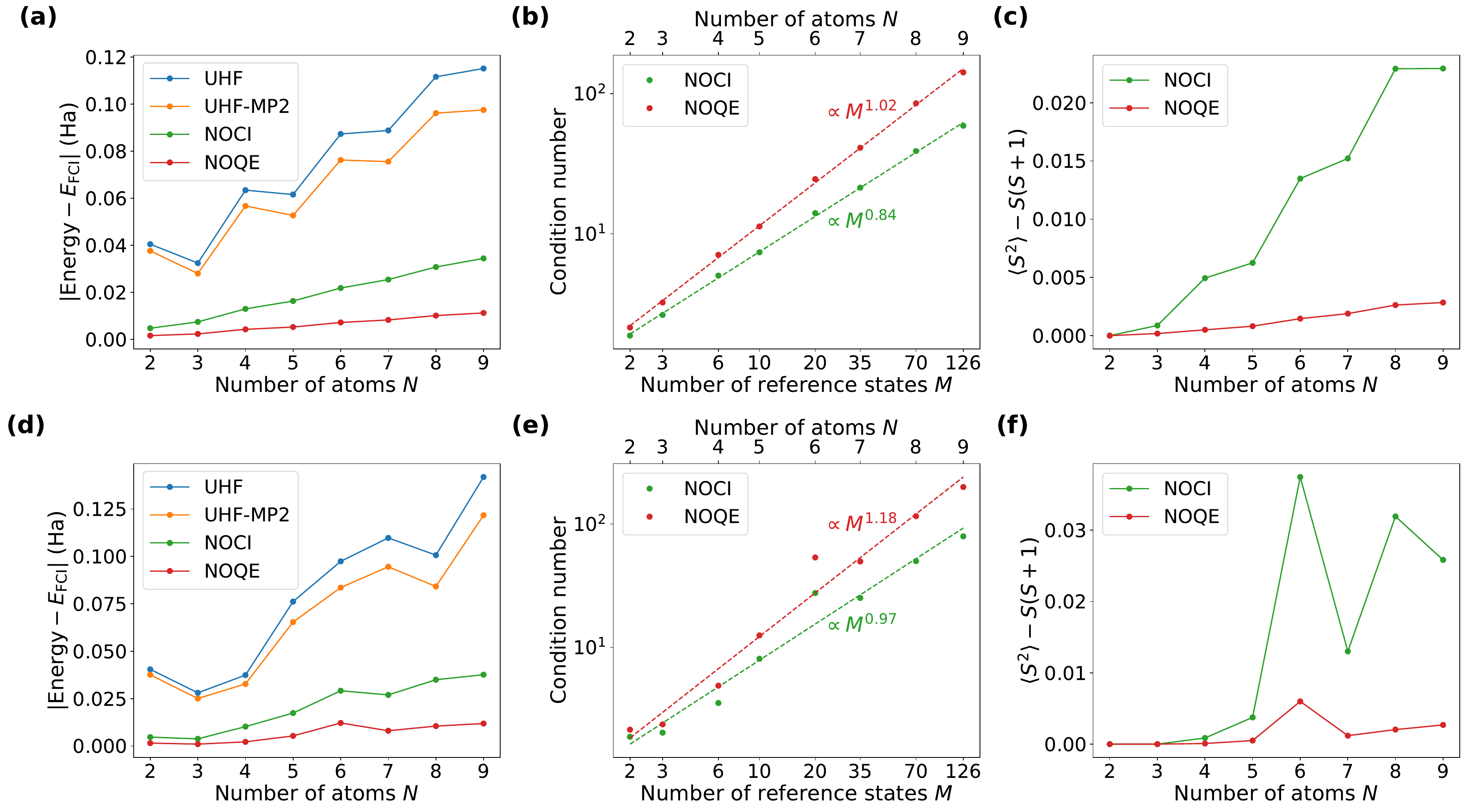}
\caption{\label{fig:h-noiseless}
Performance of noiseless NOQE for (a)--(c) linear hydrogen chains (upper panels) and (d)--(f) regular hydrogen rings (lower panels) of various numbers of atoms $N$ in the STO-3G basis, compared with classical methods such as NOCI, UHF and UHF-MP2. 
The distance between nearest-neighbor atoms is set as 1.5~\AA. 
\textbf{(a,d)} Ground-state energy error $|E-E_{\rm FCI}|$ versus $N$. 
\textbf{(b,e)} Condition number of the
unthresholded overlap matrix $\kappa(\Smat)$ versus the number of reference states $M(N)=\binom{N}{\lceil N/2\rceil}$. 
\textbf{(c,f)} Spin contamination $\langle \hat{S}^2\rangle-S(S+1)$ versus $N$, where $S=0\:(1/2)$ for even (odd) $N$. 
}
\end{figure*}

\begin{definition}[Thresholding scheme]
\label{def:thresholding-scheme}
For a Hermitian matrix pair $(\Hmat,\Smat)$ with $\Smat>0$, and for a threshold
$\vareps>0$, diagonalize the overlap matrix as
\begin{align*}
    \Smat
    &=
    V\Lambda V^\dagger,
    \\
    \Lambda
    &=
    \operatorname{diag}(\lambda_1,\ldots,\lambda_M).
\end{align*}
Let $V_{>\vareps}$ be the matrix whose columns are the eigenvectors with
$\lambda_\ell>\vareps$. Applying the threshold $\vareps$ replaces $(\Hmat,\Smat)$ by
the retained pair
\begin{align}
    \Hmat_{>\vareps}
    &=
    V_{>\vareps}^\dagger
    \Hmat
    V_{>\vareps}, \nonumber\\
    \Smat_{>\vareps}
    &=
    V_{>\vareps}^\dagger
    \Smat
    V_{>\vareps}.
    \label{eq:thresholded-pair}
\end{align}
For a family of pairs $(\Hmat^{(N)},\Smat^{(N)})$ indexed by system size $N$,
a thresholding scheme is a choice of thresholds $\{\vareps_N\}_N$, or an
equivalent rule for choosing them, applied to each pair as above.
\end{definition}

\begin{definition}[Scalable thresholding scheme]
\label{def:scalable-thresholding}
A thresholding scheme in the sense of Def.~\ref{def:thresholding-scheme}
is scalable for the $j$th eigenvalue if the retained overlap matrices satisfy
$\kappa(\Smat^{(N)}_{>\vareps_N})=\mathcal{O}(1)$ and the thresholding bias
satisfies
\begin{equation}
    \bigl|E_{j,>\vareps_N}^{(N)}-E_j^{(N)}\bigr|
    =
    \mathcal{O}(1)
    \label{eq:thresholding-bias-bound}
\end{equation}
for all $N$, where $E_j^{(N)}$ is the $j$th generalized eigenvalue in the
unthresholded reference space and $E_{j,>\vareps_N}^{(N)}$ is the corresponding
eigenvalue after thresholding. If the bound in
Eq.~\eqref{eq:thresholding-bias-bound} can be
chosen to be a prescribed tolerance $\eta$, we call the scheme
$\eta$-scalable.
\end{definition}

With these definitions, the condition-number bound can be applied to the
retained pair. This gives the following per-matrix-element measurement-cost
upper bound.

\begin{corollary}[Per-element shot cost under scalable thresholding]
\label{cor:scalable-thresholding-shot-cost}
In the setting of Theorem~\ref{thm:mathias-li-bound}, consider a family of
NOQE instances with a scalable thresholding scheme in the sense of
Def.~\ref{def:scalable-thresholding}. Dropping
the superscripts and subscripts denoting the system size $N$ for brevity, let
$M_{>\vareps}$ denote the number of retained overlap eigenvalues, and let
$X_{>\vareps}$ be the unit-column matrix that diagonalizes the retained pair
of Eq.~\eqref{eq:thresholded-pair},
\begin{equation*}
    X_{>\vareps}^\dagger
    (\Hmat_{>\vareps}+\mri\Smat_{>\vareps})
    X_{>\vareps}
    =
    \operatorname{diag}
    \bigl(\alpha_{j,>\vareps}+\mri\beta_{j,>\vareps}\bigr)_{j=1}^{M_{>\vareps}},
\end{equation*}
where $(\alpha_{j,>\vareps},\beta_{j,>\vareps})$ are the normalized
generalized eigenvalues of the retained pair, and define
\begin{equation}
    d_{\min,>\vareps}
    =
    \min_j |\alpha_{j,>\vareps}+\mri\beta_{j,>\vareps}|.
    \label{eq:dmin-thresholded}
\end{equation}
Assuming the random-matrix scaling
in Eq.~\eqref{eq:chi-rmt-scaling} and the validity condition
$\chi\kappa(\Smat_{>\vareps})<\vareps$, the number of shots per matrix element
sufficient to make the finite-shot eigenangle error of the retained pair at
most $\mathcal{E}$ is
\begin{equation}
    \Nshots =
    \mathcal{O}\!\left(\frac{M}{d_{\min,>\vareps}^2\mathcal{E}^2}\right)
    =
    \mathcal{O}\!\left(\frac{M}{\mathcal{E}^2}\right).
    \label{eq:shot-scaling}
\end{equation}
If the scheme is $\eta$-scalable, then the error of
$\widetilde{E}_{j,>\vareps}$, the finite-shot eigenvalue estimate of the
retained pair, relative to the unthresholded eigenvalue $E_j$ is at most
\begin{equation}
    \eta
    +
    \sqrt{(1+E_{j,>\vareps}^2)(1+\widetilde{E}_{j,>\vareps}^2)}\,
    \mathcal{E}.
    \label{eq:total-eigenvalue-error}
\end{equation}
\end{corollary}

\begin{proof}
See End Matter.
\end{proof}

Thus, under a scalable thresholding scheme, the worst-case Mathias-Li shot-cost upper
bound improves from $\mathcal{O}(M^3)$ to $\mathcal{O}(M)$ per matrix element.
We now give numerical evidence that such a scalable thresholding scheme exists
for hydrogen-chain and hydrogen-ring NOQE.

\begin{figure*}[t]
\centering
\includegraphics[width=\textwidth]{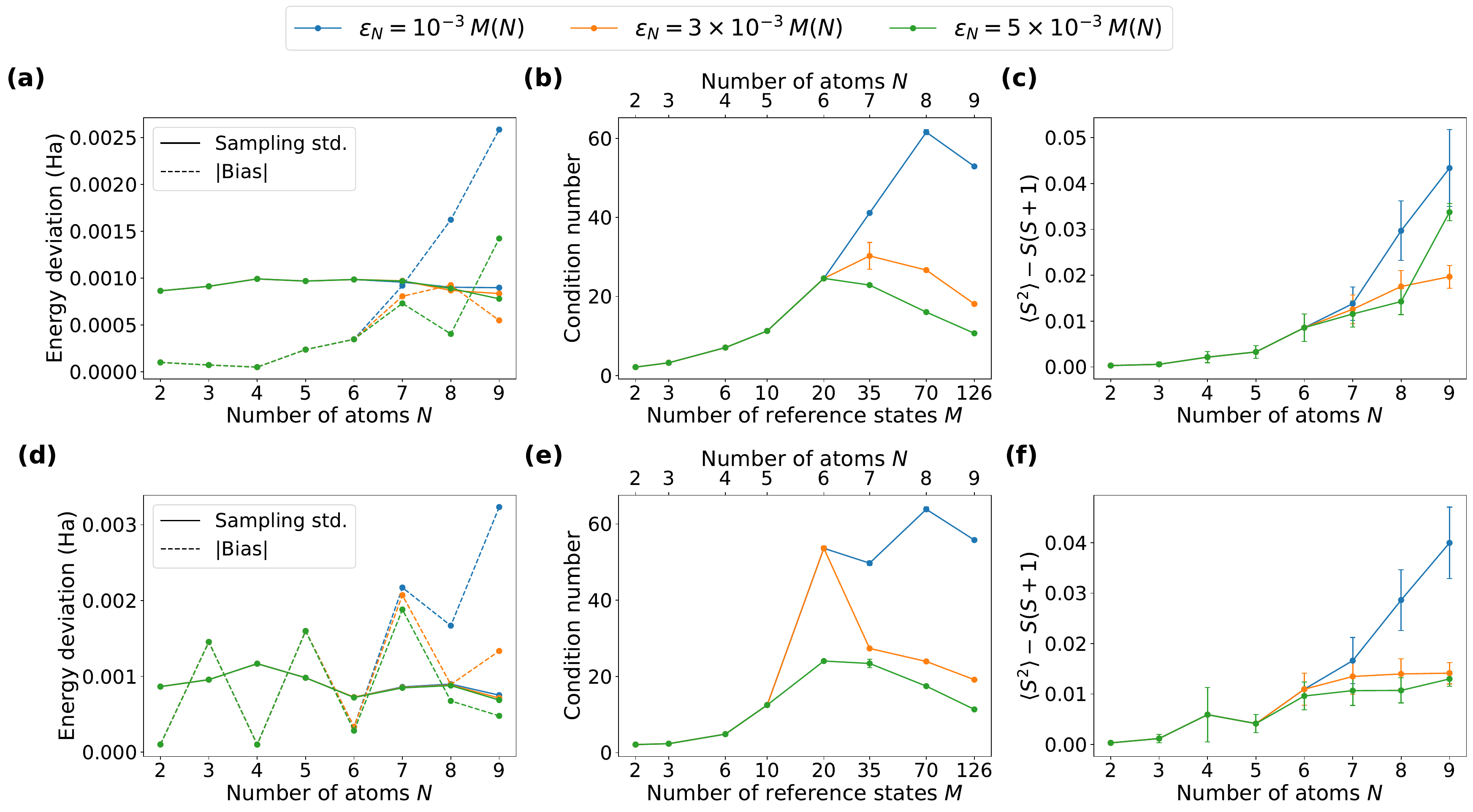}
\caption{\label{fig:h-noisy}
Performance of finite-shot NOQE with overlap thresholding for the ground state energies $E_0$ of (a)--(c) linear
hydrogen chains (upper panels) and (d)--(f) regular hydrogen rings (lower panels). Diagonal
entries of $\Hmat$ are subject to real zero-mean noise with standard deviation $\sigma=10^{-3}$~Ha, while
off-diagonal entries of $\Hmat$ and $\Smat$ are perturbed on
the upper triangle by independent zero-mean complex noise whose real and
imaginary parts each have standard deviation $\sigma =10^{-3}/\sqrt{2}$ (in Ha for $\Hmat$ and dimensionless for $\Smat$); the lower-triangular
perturbations are set by complex conjugation to preserve Hermiticity.
The results are averaged over
$300$ random realizations of noise. We consider a thresholding scheme where
$\vareps_N =c M(N), M=\binom{N}{\lceil N/2\rceil}$, with $c=10^{-3}, 3 \times 10^{-3}, 5 \times 10^{-3}$.
\textbf{(a,d)} Sampling standard deviation (solid lines) and absolute bias (dotted lines) of the noisy NOQE energy versus $N$,
where the bias is with respect to the corresponding noiseless NOQE energy.
\textbf{(b,e)} Condition number of the thresholded overlap matrix $\kappa(\Smat_{>\vareps_N})$ versus the number of reference states
$M(N)$. \textbf{(c,f)} Spin contamination metric $\langle \hat{S}^2\rangle-S(S+1)$ versus $N$, where $S=0\:(1/2)$ for even (odd) $N$.
}
\end{figure*}

\emph{Numerical evidence.} ---
We used \textsc{PySCF}~\cite{sun2020recent} for the
electronic-structure integrals and
\textsc{OpenFermion}~\cite{mcclean2020openfermion} for the second-quantized
Hamiltonians, with help from AI assistants Claude Code and ChatGPT Codex. The code that reproduces the results is publicly
available~\cite{noqe_repo}.

We test the theory on two families of molecular Hamiltonians: linear hydrogen
chains and regular hydrogen rings. For each $N$-atom system, the electronic
Hamiltonian is represented in the STO-3G basis~\cite{hehre1969self}, with
neighboring atoms separated by 1.5~\AA. In this basis the number of spin orbitals $N_{SO}$ is equal to the number of H atoms, $N_{SO}=N$.
The 1.5~\AA\:bond length places the
hydrogen systems in a strongly correlated regime~\cite{ganoe2024notion}, where unrestricted
spin-symmetry broken determinants provide a natural starting point for a compact
multireference subspace.
We consider the spin sector $m_s=0\:(1)$ for even (odd) $N$; the
corresponding target total spin is $S=0\:(1/2)$ for even (odd) $N$.

For each Hamiltonian, we determine the reference subspace by enumerating
localized spin assignments in the specified spin sector. Let $N_\alpha$ and
$N_\beta$ denote the numbers of spin-up and spin-down electrons. As
$N=N_\alpha+N_\beta$ for the hydrogen systems considered here, the spin sector fixes
$N_\alpha=\lceil N/2 \rceil$ and $N_\beta=\lfloor N/2 \rfloor$. The reference states are
generated by making localized spin assignments on the $N$ hydrogen sites:
choose a subset of $N_\alpha$ sites to carry spin-up electrons, place spin-down
electrons on the remaining sites, and use the resulting localized pattern as
the initial guess for a UHF calculation, as in Sec.~III\,D of
Ref.~\cite{baek2023say}. Varying the spin-up subset gives
$M(N)=\binom{N}{N_\alpha}=\binom{N}{\lceil N/2\rceil}$ reference states.

Figure~\ref{fig:h-noiseless} first studies the system-size dependence of noiseless
NOQE, where the projected Hamiltonian and overlap matrices are evaluated
exactly before any finite-shot measurement noise is added. In
Fig.~\ref{fig:h-noiseless}(a) and (d), we compare the NOQE ground-state energy
error with those of several classical methods: unrestricted Hartree-Fock (UHF) uses the lowest-energy determinant
among the $M$ UHF states before dressing, UHF with second-order Møller–Plesset perturbation theory (UHF-MP2), applies an MP2 correction to
this single determinant~\cite{moller1934note}, non-orthogonal configuration interaction (NOCI) diagonalizes in the
subspace of the $M$ UHF states~\cite{thom2009hartree,sundstrom2014nonorthogonal},
and NOQE diagonalizes in the subspace of the $M$ dressed UHF states, with the
dressing operator chosen as the coupled-cluster operator with MP2
amplitudes~\cite{baek2023say}. These energy panels show that NOQE remains
close to FCI across the tested systems, improving over NOCI by
roughly a factor of three as the dressing captures dynamic correlation absent from the bare UHF subspace.

Figure~\ref{fig:h-noiseless}(b) and (e) show the bare condition number of the
overlap matrix $\kappa(\Smat)$, before any thresholding, as a function of system size. Both
NOCI and NOQE exhibit a clear increasing trend, with the NOQE overlap matrices
slightly more ill-conditioned. Figures~\ref{fig:h-noiseless}(c)
and (f) plot the spin contamination, measured by $\langle \hat{S}^2\rangle-S(S+1)$, which is
substantially smaller for NOQE than for NOCI, consistent with the
spin-purification property of NOQE emphasized in Ref.~\cite{baek2023say}.

Figure~\ref{fig:h-noisy} introduces finite-shot measurement sampling noise and overlap
thresholding to assess the scalability criteria in Def.~\ref{def:scalable-thresholding}.
We set the standard deviation $\sigma$ of the independent matrix-element noise to
$10^{-3}$ (see the caption of Fig.~\ref{fig:h-noisy} for details), which places
the simulations in a regime where a scalable thresholding scheme exists and the overall
energy errors remain below chemical accuracy.
The number of shots per matrix element required to maintain this fixed element
noise may further grow modestly with system size, depending on the Pauli-word
grouping~\cite{verteletskyi2020measurement,yen2023deterministic,shlosberg2023adaptive}
or basis-rotation grouping~\cite{huggins2021efficient} used in the measurement.
We apply here three thresholding schemes $\vareps_N=cM(N)$, with
$c\in\{10^{-3},3\times10^{-3},5\times10^{-3}\}$ and solve the resulting
thresholded generalized eigenvalue problems over $300$ random realizations of
the Hermitian measurement sampling noise model to estimate the ground-state
energy $E_0$. 

First, to test whether the thresholding bias is bounded with system size, 
Fig.~\ref{fig:h-noisy}(a) and (d) report the mean bias relative to
the corresponding noiseless and unthresholded NOQE ground state energy. For an adequately chosen
thresholding scheme ($c=3\times10^{-3}$ for hydrogen chains and
$c=5\times10^{-3}$ for hydrogen rings), the bias does not show a tendency to
grow with the system size $N$ over the tested range.

Next, Fig.~\ref{fig:h-noisy}(b) and (e) show the second requirement in
Def.~\ref{def:scalable-thresholding}: after thresholding, the retained
overlap condition number $\kappa(\Smat_{>\vareps_N})$ stays controlled for both hydrogen chains and rings, in
contrast to the bare condition numbers in Fig.~\ref{fig:h-noiseless}(b) and
(e). As the system size grows, the multireference subspace may contain
directions with small overlap eigenvalues: these increase the bare condition
number but need not be important for accurately estimating the energy
eigenvalues. Together with the bias behavior in the preceding paragraph, this
provides evidence that hydrogen chains and rings admit a scalable thresholding
scheme, such that the $\mathcal{O}(M)$ measurement-cost scaling in
Corollary~\ref{cor:scalable-thresholding-shot-cost} applies. Figure~\ref{fig:h-noisy}(c)
and (f) additionally show that the spin contamination remains small after
thresholding.

We emphasize that Def.~\ref{def:scalable-thresholding} concerns the thresholding bias in
the noiseless projected matrix pair $\textbf{H}, \textbf{S}$, while the measurement-cost bound in
Eq.~\eqref{eq:shot-scaling} comes from the perturbation norm in
Eq.~\eqref{eq:chi-rmt-scaling}, which grows as $\mathcal{O}(\sqrt{M})$ at a fixed
noise level per matrix element. Figure~\ref{fig:h-noisy}(a) and (d) suggest a stronger empirical
behavior: both the sampling standard deviation and the bias remain bounded even
after finite-shot measurement noise is added. Thus, the numerical results
suggest that, in practice, the measurement cost may grow more slowly with $M$
than the linear scaling allowed by Eq.~\eqref{eq:shot-scaling}.

This is likely because Eq.~\eqref{eq:shot-scaling} is governed by the worst-case quantity
$d_{\min,>\vareps}$, whereas the actual sensitivity of an individual eigenvalue
is set by its own magnitude
$d_{j,>\vareps}=|\alpha_{j,>\vareps}+\mri\beta_{j,>\vareps}|$~\cite{mathias2004definite}.
Let $d_{0,>\vareps}$ denote this magnitude for the ground state. Since
$\alpha_{0,>\vareps}/\beta_{0,>\vareps}$ equals the ground-state energy, which
grows extensively with $N$, and
$\beta_{0,>\vareps}\geq\lambda_{\min}(\Smat_{>\vareps})=\Omega(1)$ under
scalable thresholding, $d_{0,>\vareps}$ grows with $N$, whereas
$d_{\min,>\vareps}$ is dominated by the least robust retained eigenvalue and
does not benefit from this growth. We therefore expect $d_{0,>\vareps}$,
rather than $d_{\min,>\vareps}$, to set the scale of the ground-state
eigenangle error in practice. Replacing $d_{\min,>\vareps}$ by the growing
$d_{0,>\vareps}$ in Eq.~\eqref{eq:shot-scaling} would then allow the
measurement cost to grow more slowly with $M(N)$, as the numerical results
suggest. Proving such a bound, however, requires an
assumption on the perturbation size that is stronger than the condition
$\chi \kappa(\Smat_{>\vareps})<\vareps$ of Corollary~\ref{cor:scalable-thresholding-shot-cost}: the perturbation
must be small enough that none of the less robust eigenvalues crosses the
ground-state eigenvalue~\cite{mathias2004definite}.

\emph{Discussion and outlook.} ---
Our theoretical and numerical results highlight the importance of choosing
an appropriate thresholding scheme, which must balance the improved conditioning
against the bias introduced by discarding overlap directions~\cite{devriendt2026new}. 
A complementary approach is to construct a more compact multireference subspace from the start,
thereby reducing the matrix dimension before thresholding, while preserving
the ground-state energy accuracy. Although the present examples use a subspace size
$M(N)$ that formally grows exponentially with $N$, the construction of compact multireference 
subspaces is an active area of research~\cite{giuliani2026precise}.

Although our numerical study here focuses on NOQE, the improved measurement-cost
scaling applies to any quantum subspace diagonalization method in which the
matrix elements are estimated from quantum circuits. By lowering the resource
bounds needed to estimate energies to a target accuracy, these results make
quantum computation more promising for strongly correlated quantum chemistry
problems, such as estimating spin-state energetics in transition-metal
clusters~\cite{reiher2017elucidating,vonburg2021quantum,shee2021revealing}.

\begin{acknowledgments}
The authors thank Thilo Scharnhorst, Oskar Leimkuhler, Hang Ren, Wendy Billings, Avijit Shee, Alex Krotz, Brandon Schramm, and Ayush Pancholy for discussions. 
M.K. thanks Ethan Epperly, Joerg Liesen, Chen Greif, Anne Greenbaum, and Nikhil Srivastava for their valuable feedback
during the Simons Institute ``Complexity and Linear Algebra'' program. 
This work was supported by the U.S. Department of
Energy, Office of Science, Office
of Advanced Scientific Computing Research under Award  DE-SC0025526.
\end{acknowledgments}

\section*{End Matter}

\subsection*{Proof of Lemma~\ref{lem:x-condition-number}}

Let $U$ be a unitary matrix that diagonalizes the ordinary Hermitian eigenproblem
\[
    U^\dagger \Smat^{-1/2}\Hmat\Smat^{-1/2}U
    =
    \operatorname{diag}(\lambda_1,\ldots,\lambda_M).
\]
For the Mathias-Li unit-column matrix, we write
\[
    X=\Smat^{-1/2}UD,
    \qquad
    D=\operatorname{diag}(\sqrt{\beta_1},\ldots,\sqrt{\beta_M}),
\]
where
\[
    \beta_j=\left(u_j^\dagger\Smat^{-1}u_j\right)^{-1},
    \qquad
    \alpha_j=\lambda_j\beta_j,
\]
and $u_j$ is the $j$th column of $U$. The $j$th column of $X$ is
$x_j=\sqrt{\beta_j}\Smat^{-1/2}u_j$, and the choice of $\beta_j$ gives
\[
    \|x_j\|^2
    =
    \beta_j u_j^\dagger\Smat^{-1}u_j
    =
    1,
\]
so the columns of $X$ have unit Euclidean norm. Moreover,
\[
    x_j^\dagger\Smat x_k
    =
    \sqrt{\beta_j\beta_k}\,u_j^\dagger u_k
    =
    \beta_j\delta_{jk},
\]
and
\[
    x_j^\dagger\Hmat x_k
    =
    \sqrt{\beta_j\beta_k}\,
    u_j^\dagger\Smat^{-1/2}\Hmat\Smat^{-1/2}u_k
    =
    \lambda_j\beta_j\delta_{jk}
    =
    \alpha_j\delta_{jk}.
\]
Therefore $X^\dagger(\Hmat+\mri\Smat)X
=\operatorname{diag}(\alpha_1+\mri\beta_1,\ldots,\alpha_M+\mri\beta_M)$, as
required. As $U$ is unitary,
\[
    \|X\|^2
    \leq
    \|\Smat^{-1/2}\|^2\|U\|^2\|D\|^2
    =
    \lambda_{\min}(\Smat)^{-1}\max_j\beta_j .
\]
It remains to bound $\beta_j$. From $x_j^\dagger\Smat x_j=\beta_j$ above and the
unit-column property $\|x_j\|=1$, we have $\beta_j\leq\lambda_{\max}(\Smat)$.
Combining the two estimates gives
\[
    \|X\|^2
    \leq
    \frac{\lambda_{\max}(\Smat)}{\lambda_{\min}(\Smat)}
    =
    \kappa(\Smat).
\]
This proves Lemma~\ref{lem:x-condition-number}.

\subsection*{Proof of Corollary~\ref{cor:scalable-thresholding-shot-cost}}

We first show that the validity condition
$\chi\kappa(\Smat_{>\vareps})<\vareps$ of
Corollary~\ref{cor:scalable-thresholding-shot-cost} implies that the
thresholded matrix pair $(\Hmat_{>\vareps},\Smat_{>\vareps})$ satisfies the
validity condition of Theorem~\ref{thm:mathias-li-bound}. The
finite-shot noise perturbs the retained pair by
\[
    \dHmat_{>\vareps}=V_{>\vareps}^\dagger\dHmat V_{>\vareps},
    \qquad
    \dSmat_{>\vareps}=V_{>\vareps}^\dagger\dSmat V_{>\vareps}.
\]
Since the columns of $V_{>\vareps}$ are orthonormal, $\|V_{>\vareps}\|=1$, so
$\|\dHmat_{>\vareps}\|\leq\|\dHmat\|$ and
$\|\dSmat_{>\vareps}\|\leq\|\dSmat\|$. The size of the retained
perturbation, defined as in Eq.~\eqref{eq:chi-def}, is therefore at most
$\chi$. Next, as shown in the proof of Lemma~\ref{lem:x-condition-number},
each $\beta_{j,>\vareps}$ is a Rayleigh quotient of $\Smat_{>\vareps}$ with
a unit vector, namely the corresponding column of $X_{>\vareps}$. As
every retained overlap eigenvalue exceeds the threshold,
\begin{equation}
    \min_j\beta_{j,>\vareps}
    \geq
    \lambda_{\min}(\Smat_{>\vareps})
    >
    \vareps.
    \label{eq:beta-lower-bound}
\end{equation}
Combining the assumed
$\chi\kappa(\Smat_{>\vareps})<\vareps$ with
Lemma~\ref{lem:x-condition-number} gives
\[
    \chi\|X_{>\vareps}\|^2
    \leq
    \chi\kappa(\Smat_{>\vareps})
    <
    \vareps
    <
    \min_j\beta_{j,>\vareps},
\]
which is the validity condition of Theorem~\ref{thm:mathias-li-bound}.

We next derive that the shot count in Eq.~\eqref{eq:shot-scaling} is
sufficient to reach the target eigenangle accuracy $\mathcal{E}$.
Let $\theta_{j,>\vareps}$ and $\widetilde{\theta}_{j,>\vareps}$ denote the
exact and perturbed eigenangles of the thresholded problem. Applying
Theorem~\ref{thm:mathias-li-bound} and
Lemma~\ref{lem:x-condition-number} to the retained pair yields
\[
    |\widetilde{\theta}_{j,>\vareps}-\theta_{j,>\vareps}|
    \leq
    \sin^{-1}\!\left(
        \frac{\chi\kappa(\Smat_{>\vareps})}{d_{\min,>\vareps}}
    \right).
\]
By Def.~\ref{def:scalable-thresholding},
$\kappa(\Smat_{>\vareps})=\mathcal{O}(1)$, so the eigenangle target is met
when $\chi/d_{\min,>\vareps}\leq\mathcal{E}$ up to constants. Substituting
Eq.~\eqref{eq:chi-rmt-scaling} gives the first equality of
Eq.~\eqref{eq:shot-scaling}.

It remains to show $1/d_{\min,>\vareps}=\mathcal{O}(1)$. The eigenvalues of
$\Smat_{>\vareps}$ are the retained eigenvalues of $\Smat$, which include
the largest one, so
$\lambda_{\max}(\Smat_{>\vareps})=\lambda_{\max}(\Smat)\geq 1$, where the
last inequality holds because the diagonal entries of $\Smat$ equal one.
Together with Eqs.~\eqref{eq:dmin-thresholded}
and~\eqref{eq:beta-lower-bound}, this gives
\[
    d_{\min,>\vareps}
    \geq
    \min_j\beta_{j,>\vareps}
    \geq
    \lambda_{\min}(\Smat_{>\vareps})
    \geq
    \frac{1}{\kappa(\Smat_{>\vareps})},
\]
where the final inequality follows from
$\lambda_{\max}(\Smat_{>\vareps})\geq 1$.
As $\kappa(\Smat_{>\vareps})=\mathcal{O}(1)$ by
Def.~\ref{def:scalable-thresholding}, this shows
$1/d_{\min,>\vareps}=\mathcal{O}(1)$ and hence the second equality of
Eq.~\eqref{eq:shot-scaling}.

Finally, suppose the scheme is $\eta$-scalable. Using the relations
${E_{j,>\vareps}=\cot\theta_{j,>\vareps}}$ and
${1/\sin\theta_{j,>\vareps}=\sqrt{1+E_{j,>\vareps}^2}}$, together with
their perturbed counterparts, the finite-shot eigenangle error of at most
$\mathcal{E}$ implies
\begin{align*}
    |\widetilde{E}_{j,>\vareps}-E_{j,>\vareps}|
    &=
    \frac{\sin|\widetilde{\theta}_{j,>\vareps}-\theta_{j,>\vareps}|}
         {\sin\theta_{j,>\vareps}\!\sin\widetilde{\theta}_{j,>\vareps}}
    \\
    &\leq
    \sqrt{(1+E_{j,>\vareps}^2)(1+\widetilde{E}_{j,>\vareps}^2)}\,
    \mathcal{E},
\end{align*}
where the equality follows from the cotangent difference identity
$\cot a-\cot b=\sin(b-a)/(\sin a\sin b)$ and the inequality uses
$\sin(x)\leq x$. Adding the thresholding bias
$|E_{j,>\vareps}-E_j|\leq\eta$ of Def.~\ref{def:scalable-thresholding} via
the triangle inequality gives Eq.~\eqref{eq:total-eigenvalue-error}. This
proves Corollary~\ref{cor:scalable-thresholding-shot-cost}.

We remark that the factor
$\sqrt{(1+E_{j,>\vareps}^2)(1+\widetilde{E}_{j,>\vareps}^2)}$ in
Eq.~\eqref{eq:total-eigenvalue-error} can be kept $\mathcal{O}(1)$ by
rescaling the generalized eigenvalue problem as
$(\Hmat,\Smat)\to(\Hmat/\alpha,\Smat)$, where
$\alpha\geq\|\hat{H}\|$ is a classically computable constant such as the
Pauli coefficient 1-norm $\sum_l|w_l|$. All generalized eigenvalues of the
rescaled pair satisfy $|E_j|\leq 1$, while $\Smat$, and therefore
$\kappa(\Smat_{>\vareps})$ and $\beta_{j,>\vareps}$, are unchanged. Thus,
the proof above applies verbatim with the rescaled perturbation size
$\chi=\sqrt{\|\dHmat\|^2/\alpha^2+\|\dSmat\|^2}$. The factor $1/\alpha$ is
absorbed into the constant suppressed in Eq.~\eqref{eq:chi-rmt-scaling},
which already depends on the Pauli decomposition of the Hamiltonian.

\bibliographystyle{apsrev4-2}
\bibliography{bib}

@article{tkachenko2025beyond,
  title={Beyond real: alternative unitary cluster Jastrow models for molecular electronic structure calculations on near-term quantum computers},
  author={Tkachenko, Nikolay V and Ren, Hang and Billings, Wendy M and Tomann, Rebecca and Whaley, K Birgitta and Head-Gordon, Martin},
  journal={Chemical Science},
  volume={16},
  number={47},
  pages={22299--22313},
  year={2025},
  publisher={The Royal Society of Chemistry}
}

@article{baek2023say,
  title={Say no to optimization: A nonorthogonal quantum eigensolver},
  author={Baek, Unpil and Hait, Diptarka and Shee, James and Leimkuhler, Oskar and Huggins, William J and Stetina, Torin F and Head-Gordon, Martin and Whaley, K Birgitta},
  journal={PRX Quantum},
  volume={4},
  number={3},
  pages={030307},
  year={2023},
  publisher={APS},
  doi={10.1103/PRXQuantum.4.030307}
}

@article{huggins2020nonorthogonal,
  title={A non-orthogonal variational quantum eigensolver},
  author={Huggins, William J. and Lee, Joonho and Baek, Unpil and O'Gorman, Bryan and Whaley, K. Birgitta},
  journal={New Journal of Physics},
  volume={22},
  number={7},
  pages={073009},
  year={2020},
  doi={10.1088/1367-2630/ab867b}
}

@article{huggins2021efficient,
  title={Efficient and noise resilient measurements for quantum chemistry on near-term quantum computers},
  author={Huggins, William J. and McClean, Jarrod R. and Rubin, Nicholas C. and Jiang, Zhang and Wiebe, Nathan and Whaley, K. Birgitta and Babbush, Ryan},
  journal={npj Quantum Information},
  volume={7},
  pages={23},
  year={2021},
  doi={10.1038/s41534-020-00341-7},
  eprint={1907.13117},
  archivePrefix={arXiv},
  primaryClass={quant-ph}
}

@article{verteletskyi2020measurement,
  title={Measurement optimization in the variational quantum eigensolver using a minimum clique cover},
  author={Verteletskyi, Vladyslav and Yen, Tzu-Ching and Izmaylov, Artur F},
  journal={The Journal of chemical physics},
  volume={152},
  number={12},
  year={2020},
  publisher={AIP Publishing},
  doi={10.1063/1.5141458}
}

@article{yen2023deterministic,
  title={Deterministic improvements of quantum measurements with grouping of compatible operators, non-local transformations, and covariance estimates},
  author={Yen, Tzu-Ching and Ganeshram, Aadithya and Izmaylov, Artur F.},
  journal={npj Quantum Information},
  volume={9},
  pages={14},
  year={2023},
  doi={10.1038/s41534-023-00683-y}
}

@article{shlosberg2023adaptive,
  title={Adaptive estimation of quantum observables},
  author={Shlosberg, Ariel and Jena, Andrew J and Mukhopadhyay, Priyanka and Haase, Jan F and Leditzky, Felix and Dellantonio, Luca},
  journal={Quantum},
  volume={7},
  pages={906},
  year={2023},
  publisher={Verein zur F{\"o}rderung des Open Access Publizierens in den Quantenwissenschaften},
  doi={10.22331/q-2023-01-26-906}
}

@article{grinko2021iterative,
  title={Iterative quantum amplitude estimation},
  author={Grinko, Dmitry and Gacon, Julien and Zoufal, Christa and Woerner, Stefan},
  journal={npj Quantum Information},
  volume={7},
  number={1},
  pages={52},
  year={2021},
  publisher={Nature Publishing Group UK London},
  doi={10.1038/s41534-021-00379-1}
}

@article{ren2026towards,
  title={Towards Heisenberg scaling: Measurement-efficient non-orthogonal quantum eigensolver},
  author={Ren, Hang and Zhang, Yipei and Scharnhorst, Thilo and Whaley, K. Birgitta},
  journal={arXiv preprint arXiv:2606.01589},
  year={2026},
  doi={10.48550/arXiv.2606.01589}
}

@article{romero2018strategies,
  title={Strategies for quantum computing molecular energies using the unitary coupled cluster ansatz},
  author={Romero, Jonathan and Babbush, Ryan and McClean, Jarrod R. and Hempel, Cornelius and Love, Peter J. and Aspuru-Guzik, Al{\'a}n},
  journal={Quantum Science and Technology},
  volume={4},
  number={1},
  pages={014008},
  year={2018},
  doi={10.1088/2058-9565/aad3e4}
}

@article{lee2019generalized,
  title={Generalized unitary coupled cluster wave functions for quantum computation},
  author={Lee, Joonho and Huggins, William J. and Head-Gordon, Martin and Whaley, K. Birgitta},
  journal={Journal of Chemical Theory and Computation},
  volume={15},
  number={1},
  pages={311--324},
  year={2019},
  doi={10.1021/acs.jctc.8b01004}
}

@article{matsuzawa2020jastrow,
  title={Jastrow-type decomposition in quantum chemistry for low-depth quantum circuits},
  author={Matsuzawa, Yuya and Kurashige, Yuki},
  journal={Journal of Chemical Theory and Computation},
  volume={16},
  number={2},
  pages={944--952},
  year={2020},
  doi={10.1021/acs.jctc.9b00963}
}

@article{mcardle2020quantum,
  title={Quantum computational chemistry},
  author={McArdle, Sam and Endo, Suguru and Aspuru-Guzik, Al{\'a}n and Benjamin, Simon C. and Yuan, Xiao},
  journal={Reviews of Modern Physics},
  volume={92},
  number={1},
  pages={015003},
  year={2020},
  doi={10.1103/RevModPhys.92.015003}
}

@article{hehre1969self,
  title={Self-consistent molecular-orbital methods. {I}. Use of {Gaussian} expansions of {Slater}-type atomic orbitals},
  author={Hehre, W. J. and Stewart, R. F. and Pople, J. A.},
  journal={The Journal of Chemical Physics},
  volume={51},
  number={6},
  pages={2657--2664},
  year={1969},
  doi={10.1063/1.1672392}
}

@article{moller1934note,
  title={Note on an Approximation Treatment for Many-Electron Systems},
  author={M{\o}ller, Chr. and Plesset, M. S.},
  journal={Physical Review},
  volume={46},
  number={7},
  pages={618--622},
  year={1934},
  doi={10.1103/PhysRev.46.618}
}

@article{thom2009hartree,
  title={Hartree-Fock solutions as a quasidiabatic basis for nonorthogonal configuration interaction},
  author={Thom, Alex J. W. and Head-Gordon, Martin},
  journal={The Journal of Chemical Physics},
  volume={131},
  number={12},
  pages={124113},
  year={2009},
  doi={10.1063/1.3236841}
}

@article{sundstrom2014nonorthogonal,
  title={Non-orthogonal configuration interaction for the calculation of multielectron excited states},
  author={Sundstrom, Eric J. and Head-Gordon, Martin},
  journal={The Journal of Chemical Physics},
  volume={140},
  number={11},
  pages={114103},
  year={2014},
  doi={10.1063/1.4868120}
}

@article{giuliani2026precise,
  title={Precise Quantum Chemistry calculations with few Slater Determinants},
  author={Giuliani, Clemens and Nys, Jannes and Martinazzo, Rocco and Carleo, Giuseppe and Rossi, Riccardo},
  journal={Nature Communications},
  volume={17},
  number={1},
  pages={5613},
  year={2026},
  doi={10.1038/s41467-026-70255-z}
}

@article{shee2021revealing,
  title={Revealing the nature of electron correlation in transition metal complexes with symmetry breaking and chemical intuition},
  author={Shee, Avijit and Loipersberger, Matthias and Hait, Diptarka and Lee, Joonho and Head-Gordon, Martin},
  journal={The Journal of Chemical Physics},
  volume={154},
  number={19},
  pages={194109},
  year={2021},
  doi={10.1063/5.0047386}
}

@article{ganoe2024notion,
  title={On the notion of strong correlation in electronic structure theory},
  author={Ganoe, Brad and Shee, James},
  journal={Faraday Discussions},
  volume={254},
  pages={53--75},
  year={2024},
  doi={10.1039/D4FD00066H}
}

@article{devriendt2026new,
  title={A New Angle on Quantum Subspace Diagonalization for Quantum Chemistry},
  author={De Vriendt, Xeno and Bringewatt, Jacob and Gjonbalaj, Nik O. and Ostermann, Stefan and Vodola, Davide and Borregaard, Johannes and K{\"u}hn, Michael and Yelin, Susanne F.},
  journal={arXiv preprint arXiv:2602.11985},
  year={2026},
  doi={10.48550/arXiv.2602.11985}
}

@article{kitaev1995quantum,
  title={Quantum measurements and the Abelian stabilizer problem},
  author={Kitaev, Alexei Yu.},
  journal={arXiv preprint quant-ph/9511026},
  year={1995},
  doi={10.48550/arXiv.quant-ph/9511026}
}

@article{abrams1999quantum,
  title={Quantum algorithm providing exponential speed increase for finding eigenvalues and eigenvectors},
  author={Abrams, Daniel S. and Lloyd, Seth},
  journal={Physical Review Letters},
  volume={83},
  number={24},
  pages={5162--5165},
  year={1999},
  doi={10.1103/PhysRevLett.83.5162}
}

@article{reiher2017elucidating,
  title={Elucidating reaction mechanisms on quantum computers},
  author={Reiher, Markus and Wiebe, Nathan and Svore, Krysta M. and Wecker, Dave and Troyer, Matthias},
  journal={Proceedings of the National Academy of Sciences},
  volume={114},
  number={29},
  pages={7555--7560},
  year={2017},
  doi={10.1073/pnas.1619152114}
}

@article{vonburg2021quantum,
  title={Quantum computing enhanced computational catalysis},
  author={von Burg, Vera and Low, Guang Hao and H{\"a}ner, Thomas and Steiger, Damian S. and Reiher, Markus and Roetteler, Martin and Troyer, Matthias},
  journal={Physical Review Research},
  volume={3},
  number={3},
  pages={033055},
  year={2021},
  doi={10.1103/PhysRevResearch.3.033055}
}

@article{lee2021even,
  title={Even More Efficient Quantum Computations of Chemistry Through Tensor Hypercontraction},
  author={Lee, Joonho and Berry, Dominic W. and Gidney, Craig and Huggins, William J. and McClean, Jarrod R. and Wiebe, Nathan and Babbush, Ryan},
  journal={PRX Quantum},
  volume={2},
  number={3},
  pages={030305},
  year={2021},
  doi={10.1103/PRXQuantum.2.030305}
}

@article{low2025fast,
  title={Fast Quantum Simulation of Electronic Structure by Spectral Amplification},
  author={Low, Guang Hao and King, Robbie and Berry, Dominic W and Han, Qiushi and DePrince III, A Eugene and White, Alec F and Babbush, Ryan and Somma, Rolando D and Rubin, Nicholas C},
  journal={Physical Review X},
  volume={15},
  number={4},
  pages={041016},
  year={2025},
  doi={10.1103/pb2g-j9cw}
}

@article{mcclean2017hybrid,
  title={Hybrid quantum-classical hierarchy for mitigation of decoherence and determination of excited states},
  author={McClean, Jarrod R. and Kimchi-Schwartz, Mollie E. and Carter, Jonathan and de Jong, Wibe A.},
  journal={Physical Review A},
  volume={95},
  number={4},
  pages={042308},
  year={2017},
  publisher={APS},
  doi={10.1103/PhysRevA.95.042308}
}

@article{colless2018computation,
  title={Computation of molecular spectra on a quantum processor with an error-resilient algorithm},
  author={Colless, J. I. and Ramasesh, V. V. and Dahlen, D. and Blok, M. S. and Kimchi-Schwartz, M. E. and McClean, J. R. and Carter, J. and de Jong, W. A. and Siddiqi, I.},
  journal={Physical Review X},
  volume={8},
  number={1},
  pages={011021},
  year={2018},
  publisher={APS},
  doi={10.1103/PhysRevX.8.011021}
}

@article{cortes2022fast,
  title={Fast-forwarding quantum simulation with real-time quantum {Krylov} subspace algorithms},
  author={Cortes, Cristian L and DePrince III, A Eugene and Gray, Stephen K},
  journal={Physical Review A},
  volume={106},
  number={4},
  pages={042409},
  year={2022},
  publisher={APS},
  doi={10.1103/PhysRevA.106.042409}
}

@article{stair2020multireference,
  title={A multireference quantum {Krylov} algorithm for strongly correlated electrons},
  author={Stair, Nicholas H. and Huang, Renke and Evangelista, Francesco A.},
  journal={Journal of Chemical Theory and Computation},
  volume={16},
  number={4},
  pages={2236--2245},
  year={2020},
  publisher={ACS Publications},
  doi={10.1021/acs.jctc.9b01125}
}

@article{shen2023real,
  title={Real-time {Krylov} theory for quantum computing algorithms},
  author={Shen, Yizhi and Klymko, Katherine and Sud, James and Williams-Young, David B. and de Jong, Wibe A. and Tubman, Norm M.},
  journal={Quantum},
  volume={7},
  pages={1066},
  year={2023},
  publisher={Verein zur F{\"o}rderung des Open Access Publizierens in den Quantenwissenschaften},
  doi={10.22331/q-2023-07-25-1066}
}

@article{seki2021quantum,
  title={Quantum Power Method by a Superposition of Time-Evolved States},
  author={Seki, Kazuhiro and Yunoki, Seiji},
  journal={PRX Quantum},
  volume={2},
  number={1},
  pages={010333},
  year={2021},
  doi={10.1103/PRXQuantum.2.010333}
}

@article{parrish2019quantum,
  title={Quantum Filter Diagonalization: Quantum Eigendecomposition without Full Quantum Phase Estimation},
  author={Parrish, Robert M. and McMahon, Peter L.},
  journal={arXiv preprint arXiv:1909.08925},
  year={2019},
  doi={10.48550/arXiv.1909.08925}
}

@article{bespalova2021hamiltonian,
  title={Hamiltonian Operator Approximation for Energy Measurement and Ground-State Preparation},
  author={Bespalova, Tatiana A. and Kyriienko, Oleksandr},
  journal={PRX Quantum},
  volume={2},
  number={3},
  pages={030318},
  year={2021},
  doi={10.1103/PRXQuantum.2.030318}
}

@article{patel2026quantum,
  title={Quantum seniority-based subspace expansion: Linear combinations of short-circuit unitary transformations for the electronic structure problem},
  author={Patel, Smik and Jayakumar, Praveen and Huang, Rick and Zeng, Tao and Izmaylov, Artur F.},
  journal={Journal of Chemical Theory and Computation},
  volume={22},
  number={8},
  pages={3937--3949},
  year={2026},
  doi={10.1021/acs.jctc.6c00017}
}

@article{leimkuhler2025quantum,
  title={A quantum eigenvalue solver based on tensor networks},
  author={Leimkuhler, Oskar and Whaley, K. Birgitta},
  journal={npj Quantum Information},
  volume={11},
  pages={184},
  year={2025},
  doi={10.1038/s41534-025-01128-4}
}

@article{leimkuhler2025exponential,
  title={Exponential quantum speedups for near-term molecular electronic structure methods},
  author={Leimkuhler, Oskar and Whaley, K. Birgitta},
  journal={arXiv preprint arXiv:2503.21041},
  year={2025},
  doi={10.48550/arXiv.2503.21041}
}

@incollection{vershynin2012nonasymptotic,
  title={Introduction to the non-asymptotic analysis of random matrices},
  author={Vershynin, Roman},
  booktitle={Compressed Sensing: Theory and Applications},
  pages={210--268},
  publisher={Cambridge University Press},
  year={2012}
}

@techreport{mathias2004definite,
  title={The definite generalized eigenvalue problem: A new perturbation theory},
  author={Mathias, Roy and Li, Chi-Kwong},
  type={T-NAREP No.},
  number={457},
  institution={Manchester Centre for Computational Mathematics},
  address={Manchester, UK},
  month=oct,
  note={{Oct. 2004}},
  year={2004}
}

@article{epperly2022theory,
  title={A theory of quantum subspace diagonalization},
  author={Epperly, Ethan N and Lin, Lin and Nakatsukasa, Yuji},
  journal={SIAM Journal on Matrix Analysis and Applications},
  volume={43},
  number={3},
  pages={1263--1290},
  year={2022},
  publisher={SIAM},
  doi={10.1137/21M145954X}
}

@article{lee2024sampling,
  title={Sampling error analysis in quantum krylov subspace diagonalization},
  author={Lee, Gwonhak and Lee, Dongkeun and Huh, Joonsuk},
  journal={Quantum},
  volume={8},
  pages={1477},
  year={2024},
  publisher={Verein zur F{\"o}rderung des Open Access Publizierens in den Quantenwissenschaften},
  doi={10.22331/q-2024-09-19-1477}
}

@article{kirby2024analysis,
  title={Analysis of quantum Krylov algorithms with errors},
  author={Kirby, William},
  journal={Quantum},
  volume={8},
  pages={1457},
  year={2024},
  publisher={Verein zur F{\"o}rderung des Open Access Publizierens in den Quantenwissenschaften},
  doi={10.22331/q-2024-08-29-1457}
}

@article{ren2026error,
  title={Error-mitigated nonorthogonal quantum eigensolver via shadow tomography},
  author={Ren, Hang and Zhang, Yipei and Billings, Wendy M and Tomann, Rebecca and Tkachenko, Nikolay V and Kang, Mingyu and Head-Gordon, Martin and Whaley, K Birgitta},
  journal={Physical Review Research},
  volume={8},
  number={1},
  pages={013268},
  year={2026},
  publisher={APS},
  doi={10.1103/7c5b-3v56}
}

@article{sun2020recent,
  title={Recent developments in the PySCF program package},
  author={Sun, Qiming and Zhang, Xing and Banerjee, Samragni and Bao, Peng and Barbry, Marc and Blunt, Nick S and Bogdanov, Nikolay A and Booth, George H and Chen, Jia and Cui, Zhi-Hao and others},
  journal={The Journal of chemical physics},
  volume={153},
  number={2},
  pages={024109},
  year={2020},
  publisher={AIP Publishing},
  doi={10.1063/5.0006074}
}

@article{mcclean2020openfermion,
  title={OpenFermion: the electronic structure package for quantum computers},
  author={McClean, Jarrod R and Rubin, Nicholas C and Sung, Kevin J and Kivlichan, Ian D and Bonet-Monroig, Xavier and Cao, Yudong and Dai, Chengyu and Fried, E Schuyler and Gidney, Craig and Gimby, Brendan and others},
  journal={Quantum Science \& Technology},
  volume={5},
  number={3},
  pages={034014},
  year={2020},
  publisher={IOP Publishing},
  doi={10.1088/2058-9565/ab8ebc}
}

@misc{noqe_repo,
  author={Kang, Mingyu},
  title={{Code for ``Improved Measurement Cost Scaling in the Nonorthogonal Quantum Eigensolver''}},
  year={2026},
  howpublished={\url{https://github.com/mkangquantum/noqe.git}},
  note={GitHub repository}
}

\end{document}